\documentclass[12pt,a4paper]{article}

\usepackage{amsfonts,amsmath,amssymb,amsthm,latexsym}
\usepackage[all]{xy}

\newtheorem{proposition}{Proposition}[section]
\newtheorem{lemma}[proposition]{Lemma}
\newtheorem{corollary}[proposition]{Corollary}
\newtheorem{theorem}[proposition]{Theorem}

\theoremstyle{definition}
\newtheorem{definition}[proposition]{Definition}

\theoremstyle{remark}

\newcommand{\thlabel}[1]{\label{th:#1}}
\newcommand{\thref}[1]{Theorem~\ref{th:#1}}
\newcommand{\selabel}[1]{\label{se:#1}}
\newcommand{\seref}[1]{Section~\ref{se:#1}}
\newcommand{\lelabel}[1]{\label{le:#1}}
\newcommand{\leref}[1]{Lemma~\ref{le:#1}}

\newcommand{\colabel}[1]{\label{co:#1}}
\newcommand{\coref}[1]{Corollary~\ref{co:#1}}

\newcommand{\delabel}[1]{\label{de:#1}}

\title{\textbf{Mixed Paraparticles, Colors, Braidings and a new class of Realizations for Lie superalgebras}}

\author{\textsc{K. Kanakoglou} \\ kanakoglou@hotmail.com {\small and}
kanakoglou@ifm.umich.mx \\
\small{Instituto de F\'{\i}sica y Matem\'{a}ticas - \textsc{Ifm}},  \\
\small{Universidad Michoacana de San Nicolas de Hidalgo - \textsc{Umsnh}} \\
\small{Edificio C-3, Cd. Universitaria, CP 58040},
\small{Morelia, Michoac\'{a}n, \textsc{Mexico}} \\
\footnotesize{and:} \\
\small{Phys. Dept.},
\small{Aristotle University of Thessaloniki - \textsc{Auth}} \\
\small{54124 Thessaloniki, \textsc{Greece}} \\ \\
\textsc{C. Daskaloyannis} - daskalo@math.auth.gr \\ \small{Math. Dept.},
\small{Aristotle University of Thessaloniki - \textsc{Auth}} \\
\small{54124 Thessaloniki, \textsc{Greece}} \\  \\
\textsc{A. Herrera-Aguilar} - herrera@ifm.umich.mx \\
\small{Instituto de F\'{\i}sica y Matem\'{a}ticas - \textsc{Ifm}}, \\
\small{Universidad Michoacana de San Nicol\'{a}s de Hidalgo - \textsc{Umsnh}} \\
\small{Edificio C-3, Cd. Universitaria, CP 58040},
\small{Morelia, Michoac\'{a}n, \textsc{Mexico}}}

\date{}

\begin{document}

\maketitle

\begin{abstract}
A rigorous algebraic description of the notion of realization, specialized in the
case of Lie superalgebras is given. The idea of the Relative Parabose set $P_{BF}$ is recalled together with some
recent developments and its braided group structure is established together with an extended discussion
of its ($\mathbb{Z}_{2} \times \mathbb{Z}_{2}$)-grading. The final result of the paper
employs $P_{BF}$ in order to realize an arbitrary Lie superalgebra. It is
furthermore shown that the constructed realization is a $\mathbb{Z}_{2}$-graded Hopf algebra homomorphism.
\end{abstract}

\section{Introduction: the concept of realization for a Lie superalgebra} \selabel{realizationdef}

The concept of ``realizing'' an algebra of operators via bosons or fermions has proved a very fruitful
approach in both the levels of first and second quantization theories. The relevant examples are numerous:
In the review article \cite{KlMa} and the references therein one can find a host of applications in models of
nuclear physics, in \cite{Ca} such techniques are applied in solid state physics problems, while a review
in various areas of physics can be found in \cite{Iache} -just to mention a few. Such techniques are equally
useful in problems of purely mathematical character: The famous Jordan-Schwinger map, which was introduced in 1935
by Pascual Jordan \cite{Jor} for any Lie algebra and was ``rediscovered'' by Shwinger in 1953 for the special example
of $su(2)$, led to the construction of the symmetric and the antisymmetric representations of Lie algebras.

The idea of realizing an algebra of operators in terms of an other algebra (usually bosonic or fermionic algebras)
roughly amounts in constructing an homomorphism between our initial algebra and the ``target'' algebra. In this
way, the relations of the initial algebra can be directly reproduced by the relations of the target algebra (which are either
easier to handle or more well known). Moreover, the representations of the target algebra give rise -in a straightforward way-
to representations of the initial algebra, enabling us to gain information about the spectrum, the eigenstates and thus the
expected values of various physical quantities. The dynamical variables of any physical theory frequently have the
structure of the universal enveloping algebra (UEA) of some Lie algebra and this is exactly what initiated the interest
in realizations of Lie algebras. Since the introduction of the idea of supersymmetry, Lie superalgebras and their UEAs
have dramatically come into play. Nowadays the Lie superalgebraic structure can be found throughout models of quantum field theory,
elementary particle physics, solid state physics etc.

Let $L = L_{0} \oplus L_{1}$ be a Lie superalgebra (equivalently: a $\mathbb{Z}_{2}$-graded Lie algebra) with $L_{0}$ its even
and $L_{1}$ its odd part, $\langle .., .. \rangle : L \times L \rightarrow L$ its non-associative
multiplication, $\mathbb{U}(L)$ its universal enveloping algebra UEA (which is well known to be an associative superalgebra or
equivalently a  $\mathbb{Z}_{2}$-graded associative algebra) and $i_{\mathbb{U}(L)} : L \rightarrow \mathbb{U}(L)$ the canonical
injection of $L$ into $\mathbb{U}(L)$. We know that $i_{\mathbb{U}(L)}$ is a linear map of super vector spaces or equivalently
an even linear map or equivalently an homogeneous linear map of degree $0$. We recall here that in $\mathbb{U}(L)$ we will have the
relations
\begin{equation} \label{Liesupercom}
\langle x, y \rangle = xy - (-1)^{deg(x)deg(y)}yx
\end{equation}
for any $x,y$ homogeneous elements of $L$ or equivalently homogeneous generators of $\mathbb{U}(L)$.

If $\mathcal{A}_{super}$ is any associative superalgebra and $J_{L} : L \rightarrow \mathcal{A}_{super}$ an even linear map which
furthermore satisfies
\begin{equation} \label{suffcondreal}
J_{L}( \langle x, y \rangle ) = J_{L}(x)J_{L}(y) - (-1)^{deg(x)deg(y)} J_{L}(y)J_{L}(x)
\end{equation}
for any homogeneous elements $x, y \in L$, then there exists a unique homomorphism of associative superalgebras (equivalently: an
even homomorphism of associative algebras or: an homogeneous homomorphism of associative algebras of degree $0$)
$J : \mathbb{U}(L) \rightarrow \mathcal{A}_{super}$ which extends the linear map $J_{L}$, in such a way that the following diagram
becomes commutative
\begin{equation}
\begin{array}{ccc}
\xymatrix{L \ar[rr]^{J_{L}} \ar[d]_{i_{\mathbb{U}(L)}} & & \mathcal{A}_{super} \\
 \mathbb{U}(L) \ar@{.>}[urr]_{\exists ! \ J} & & }
  & \textrm{ή ισοδύναμα :} & J \circ i_{\mathbb{U}(L)} = J_{L}
\end{array}
\end{equation}
$J$ is uniquely determined by its values on  the generators of $\mathbb{U}(L)$ which is the values of $J_{L}$ on the elements of $L$.
We can now state the following

\begin{definition} \delabel{realization}
Under the above conditions, the unique homomorphism of associative algebras $J : \mathbb{U}(L) \rightarrow \mathcal{A}_{super}$
(or sometimes by abuse of terminology even the linear map $J_{L} : L \rightarrow \mathcal{A}_{super}$)
will be called a \emph{realization} of the Lie superalgebra $L$ in terms of the elements of the algebra $\mathcal{A}_{super}$.
\end{definition}

If furthermore, the associative superalgebra $\mathcal{A}_{super}$ is an algebra consisting of bosons or fermions or some mixture of
them  (with any gradation satisfying the above) then the homomorphism $J$ will be called a realization of the Lie superalgebra $L$
with bosons or fermions or with particles. Correspondingly, if $\mathcal{A}_{super}$ is an algebra consisting of parabosons or
parafermions or some mixture of them (with the necessary gradation) then we will be speaking about a realization of the
Lie superalgebra $L$ with parabosons or parafermions or paraparticles.

Although Lie superalgebra realizations is a much more recent topic than Lie algebra realizations, there is already a significant
number of references on this subject. Far from trying to present an exhaustive list, we feel it's worth underlining the readers attention to
the general works \cite{Hong, Pal17, ChangS, Tang} as well as the review \cite{FrSorSc1} and the references therein. In these
works and -to the best of our knowledge- in any similar reference in the literature, the target algebra is almost always some
mixture of bosons and fermions (ordinary or deformed). In some cases, algebras of differential operators -acting on suitable polynomial
spaces- have also been used \cite{Hong}. We can thus speak about \emph{particle realizations} or \emph{differential realizations} of Lie superalgebras.
The novel thing in this article will be the use of a mixed paraparticle system as a target algebra. Consequently we will deal with
\emph{paraparticle realizations} of an arbitrary Lie superalgebra.

In the next section (\seref{relparabset}) we will present our target algebra which will be the Relative parabose set $P_{BF}$
and recall some recent developments about its structure. Motivated by
these developments, we will make an extended discussion of its ($\mathbb{Z}_{2} \times \mathbb{Z}_{2}$)-graded structure and we will prove two lemmas:
one of them establishing a braided group structure for $P_{BF}$ and another one specifying suitable subalgebras that will be of use in the last section.

In the last section (\seref{newresults}) the Relative parabose set will be used in order for a new family of realizations to be constructed for an
arbitrary Lie superalgebra of either finite or infinite dimension. The only assumption will be the existence of a finite dimensional graded representation
(i.e.: a graded-matrix representation) for the Lie superalgebra. We will furthermore show that the constructed realization is an homomorphism between
braided groups and more specifically that it is an homomorphism of super-Hopf algebras. The section will close by studying some consequences
which show that the constructed realization generalizes previous results of ours and of other authors as well.

\section{Relative parabose set $P_{BF}$: a $\theta$-braided group} \selabel{relparabset}

The \emph{Relative Parabose Set} has been introduced by Greenberg and Messiah at their seminal paper \cite{GreeMe}.
It has been historically the only -together with a couple of other models introduced in the same paper-
attepmt for a mixture of parabosonic and parafermionic degrees of freedom. However, the notation used in \cite{GreeMe}
to describe the algebra is compact and rather awkward. We present here the Relative Parabose Set algebra (we are going to
denote it by $P_{BF}$ from now on) in terms of generators and relations, and we will write down the relations in detail for
convenience:

$P_{BF}$ is generated -as an associative algebra- by the generators $B_{i}^{+}$, $B_{j}^{-}$, $F_{k}^{+}$, $F_{l}^{-}$
for all values $i,j,k,l = 1, 2, ...$.
In other words it will be an algebra with an infinite number of generators in the general case. If we consider finite number
of generators, for example $m$ parabosonic generators $B_{i}^{\pm}$ ($i = 1, 2, ..., m$) and $n$ parafermionic generators
$F_{j}^{\pm}$ ($j = 1, 2, ..., n$) we will use the notation $P_{BF}^{(m,n)}$.
The relations satisfied by the above generators will be:

The usual relations of the parabosonic and the parafermionic algebras which can be compactly summarized as
\small{\begin{equation} \label{parab-paraf}
\begin{array}{c}
\big[ \{ B_{i}^{\xi},  B_{j}^{\eta}\}, B_{k}^{\epsilon}  \big] = (\epsilon - \eta)\delta_{jk}B_{i}^{\xi} + (\epsilon -
 \xi)\delta_{ik}B_{j}^{\eta}    \\
    \\
\big[ [ F_{i}^{\xi},  F_{j}^{\eta} ], F_{k}^{\epsilon}  \big] = \frac{1}{2}(\epsilon - \eta)^{2} \delta_{jk}F_{i}^{\xi} - \frac{1}{2}(\epsilon -
 \xi)^{2} \delta_{ik}F_{j}^{\eta}   \\
 \end{array}
\end{equation}}
for all values $i, j, k = 1, 2, ..., $ and $\xi, \eta, \epsilon = \pm$ together with the mixed trilinear relations which characterize the
relative parabose set
\small{\begin{equation} \label{parabparaf}
\begin{array} {ccc}
\big[ \{ B_{k}^{+}, B_{l}^{-} \}, F_{m}^{-} \big] = 0  & & \big[ [ F_{k}^{+}, F_{l}^{-} ], B_{m}^{-} \big] = 0     \\
    &     &              \\
\big[ \{ B_{k}^{-}, B_{l}^{-} \}, F_{m}^{-} \big] = 0  & &   \big[ [ F_{k}^{-}, F_{l}^{-} ], B_{m}^{-} \big] = 0    \\
   &       &          \\
\big[ \{ B_{k}^{+}, B_{l}^{+} \}, F_{m}^{-} \big] = 0  &  &  \big[ [ F_{k}^{+}, F_{l}^{+} ], B_{m}^{-} \big] = 0   \\
   &   &                   \\
   &     &              \\
\big[ \{ F_{m}^{-}, B_{k}^{+} \}, B_{l}^{-} \big] = -2 \delta_{kl} F_{m}^{-}  &  &  \{ \{ B_{m}^{-}, F_{k}^{+} \}, F_{l}^{-} \} = 2 \delta_{kl} B_{m}^{-}   \\
   &   &                      \\
\big[ \{ B_{l}^{-}, F_{m}^{-} \}, B_{k}^{+} \big] =  2 \delta_{kl} F_{m}^{-}  & & \{ \{ F_{l}^{-}, B_{m}^{-} \}, F_{k}^{+} \} =  2 \delta_{kl} B_{m}^{-} \\
   &                     \\
\big[ \{ B_{k}^{-}, B_{l}^{+} \}, F_{m}^{+} \big] = 0  & &  \big[ [ F_{k}^{-}, F_{l}^{+} ], B_{m}^{+} \big] = 0    \\
   &     &                 \\
\big[ \{ F_{m}^{+}, B_{k}^{-} \}, B_{l}^{+} \big] = 2 \delta_{kl} F_{m}^{+}  & &  \{ \{ B_{m}^{+}, F_{k}^{-} \}, F_{l}^{+} \} = 2 \delta_{kl} B_{m}^{+}  \\
    &   &                       \\
\big[ \{ B_{l}^{+}, F_{m}^{+} \}, B_{k}^{-} \big] = -2 \delta_{kl} F_{m}^{+}  & &   \{ \{ F_{l}^{+}, B_{m}^{+} \}, F_{k}^{-} \} = 2 \delta_{kl} B_{m}^{+}  \\
   & &                       \\
   &     &              \\
\big[ \{ F_{m}^{-}, B_{k}^{-} \}, B_{l}^{-} \big] = 0   &  &  \{ \{ B_{m}^{-}, F_{k}^{-} \}, F_{l}^{-} \} = 0     \\
    &   &                           \\
\big[ \{ B_{k}^{+}, B_{l}^{+} \}, F_{m}^{+} \big] = 0   &  &  \big[ [ F_{k}^{+}, F_{l}^{+} ], B_{m}^{+} \big] = 0   \\
   &          &                     \\
\big[ \{ F_{m}^{+}, B_{k}^{+} \}, B_{l}^{+} \big] = 0   & &  \{ \{ B_{m}^{+}, F_{k}^{+} \}, F_{l}^{+} \} = 0     \\
   &         &                   \\
   &     &              \\
\big[ \{ F_{m}^{-}, B_{k}^{+} \}, B_{l}^{+} \big] = 0   &  &  \{ \{ B_{m}^{-}, F_{k}^{+} \}, F_{l}^{+} \} = 0  \\
    &   &                          \\
\big[ \{ B_{k}^{-}, B_{l}^{-} \}, F_{m}^{+} \big] = 0    & &    \big[ [ F_{k}^{-}, F_{l}^{-} ], B_{m}^{+} \big] = 0   \\
    &     &                            \\
\big[ \{ F_{m}^{+}, B_{k}^{-} \}, B_{l}^{-} \big] = 0   &  &  \{ \{ B_{m}^{+}, F_{k}^{-} \}, F_{l}^{-} \} = 0
\end{array}
\end{equation}}
for all values $i, j, k = 1, 2, ..., $. One can easily observe that the relations \eqref{parab-paraf} involve only the parabosonic degrees of freedom and the
parafermionic degrees of freedom seperately while the relations \eqref{parabparaf} mix the parabosonic with the
parafermionic degrees of freedom according to the recipe first proposed by Greenberg and Messiah in \cite{GreeMe}. In all the above and in what follows,
we use the notation $[x, y]$ (i.e.: the ``commutator'') to imply the expression $xy-yx$ and the notation $\{x, y \}$ (i.e.: the ``anticommutator'')
to imply the expression $xy+yx$, for $x$ and $y$ any elements of the algebra $P_{BF}$.

Recently, in a series of very interesting papers \cite{Ya1, Ya2}, the authors make some very important observations and remarks about the graded
structure and  properties of the relative parabose set $P_{BF}$. Being theoretical physicists, they seem to be mainly interested in analyzing the various
symmetries and supersymmetries of $P_{BF}$ rather than focusing on a precise formulation of their important mathematical results. Before proceeding
to a rigorous statement of their results employing mainstream mathematical terminology, we find it useful for convenience of the reader to briefly recall
the idea of $\theta$-colored $G$-graded Lie algebra (for brevity: $(G, \theta)$-Lie algebra) \cite{Kakw, Scheu},
for the special case for which $G$ is the finite abelian group $\mathbb{Z}_{2} \times \mathbb{Z}_{2}$.

A $\theta$-colored ($\mathbb{Z}_{2} \times \mathbb{Z}_{2}$)-graded Lie algebra (or: ($\mathbb{Z}_{2} \times \mathbb{Z}_{2}, \theta$)-Lie algebra)
consists of the following data:
\begin{enumerate}
\item a ($\mathbb{Z}_{2} \times \mathbb{Z}_{2}$)-graded vector space $L$ (equivalently: a representation of the group Hopf algebra
$\mathbb{C}(\mathbb{Z}_{2} \times \mathbb{Z}_{2})$) thus:
\begin{equation} \label{gradedvsLie}
L = L_{00} \oplus L_{01} \oplus L_{10} \oplus L_{11}
\end{equation}
\item a non-associative multiplication $\langle .., .. \rangle : L \times L \rightarrow L$, on $L$ respecting the gradation, and
\item a function $\theta:\big( \mathbb{Z}_{2} \times \mathbb{Z}_{2} \big) \times \big( \mathbb{Z}_{2} \times \mathbb{Z}_{2} \big)
\rightarrow \mathbb{C}^{*}$ which is usually called a \emph{color function}  or a \emph{commutation factor} on $\mathbb{Z}_{2} \times \mathbb{Z}_{2}$
\end{enumerate}
The above data must be subject to the following set of axioms:
\begin{itemize}
\item $\langle L_{a}, L_{b} \rangle \subseteq L_{a+b}$, $ \ $ $a,b \in \mathbb{Z}_{2} \times \mathbb{Z}_{2}$

\item $\theta$-\textbf{braided} (($\mathbb{Z}_{2} \times \mathbb{Z}_{2}$)-graded) antisymmetry:
$\langle x, y \rangle = - \theta(a,b) \langle y, x \rangle$

\item $\theta$-\textbf{braided} (($\mathbb{Z}_{2} \times \mathbb{Z}_{2}$)-graded) Jacobi identity:
\\
$\theta(c,a) \langle x, \langle y, z \rangle \rangle + \theta(b,c)
\langle z, \langle x, y \rangle \rangle +
  \theta(a,b) \langle y, \langle z, x \rangle \rangle = 0$

\item $\theta:\big( \mathbb{Z}_{2} \times \mathbb{Z}_{2} \big)
\times \big( \mathbb{Z}_{2} \times \mathbb{Z}_{2} \big)
\rightarrow \mathbb{C}^{*} \rightsquigarrow$ \textbf{color
function}
              \begin{itemize}
              \item[-] $\theta(a+b, c) = \theta(a,c) \theta(b,c)$
              \item[-] $\theta(a, b+c) = \theta(a,b) \theta(a,c)$
              \item[-] $\theta(a,b) \theta(b,a) = 1$
              \end{itemize}
\end{itemize}
$\forall$ homogeneous $x \in L_{a}$, $y \in L_{b}$ and
$\forall$ $a = (a_{1}, a_{2}), b = (b_{1}, b_{2}) \in \mathbb{Z}_{2} \times \mathbb{Z}_{2}$.

It is easy for one to observe that the axioms imposed on the function $\theta:\big( \mathbb{Z}_{2} \times \mathbb{Z}_{2} \big)
\times \big( \mathbb{Z}_{2} \times \mathbb{Z}_{2} \big) \rightarrow \mathbb{C}^{*}$ imply that it is a skew-symmetric bicharacter
on $\mathbb{Z}_{2} \times \mathbb{Z}_{2}$ which has been shown to be equivalent \cite{Mon, Scheu1} to a triangular universal $R$-matrix on the
group Hopf algebra $\mathbb{C}(\mathbb{Z}_{2} \times \mathbb{Z}_{2})$. We thus conclude that by definition, the notion of
$\theta$-colored ($\mathbb{Z}_{2} \times \mathbb{Z}_{2}$)-graded Lie algebra implicitly presupposes the existence of a triangular
structure for the group Hopf algebra $\mathbb{C}(\mathbb{Z}_{2} \times \mathbb{Z}_{2})$ which finally entails a symmetric braiding
in the monoidal category ${}_{\mathbb{C}(\mathbb{Z}_{2} \times \mathbb{Z}_{2})}\mathcal{M}$ of its representations. This last remark,
fully justifies the use of the term braided, in both the antisymmetry property and the generalized Jacobi identity included in the
defining axioms of the $\theta$-colored ($\mathbb{Z}_{2} \times \mathbb{Z}_{2}$)-graded Lie algebra.

We furthermore recall (see \cite{Kakw, Scheu} for a more detailed description), that the universal enveloping algebra UEA of $L$ is
denoted by $\mathbb{U}(L)$ and it is a ($\mathbb{Z}_{2} \times \mathbb{Z}_{2}$)-graded associative algebra. It will be generated (as an
associative algebra) by any linearly independent set of homogeneous elements of $L$ or in other words any homogeneous basis of
$L$ will provide a set of homogeneous generators of $\mathbb{U}(L)$. Finally, in $\mathbb{U}(L)$ we will have the relations
\begin{equation} \label{braidZ22Liecom}
\langle x, y \rangle = xy - \theta(deg(x),deg(y))yx
\end{equation}
for any $x,y$ homogeneous elements of $L$ (or equivalently: any homogeneous generators of $\mathbb{U}(L)$)
and $deg(x), deg(y) \in  \mathbb{Z}_{2} \times \mathbb{Z}_{2}$.

We can now proceed to summarizing the mathematical results of \cite{Ya1, Ya2} in the following
\begin{theorem} \thlabel{gradstrRelParabSet}
The relative parabose set $P_{BF}$ is the universal enveloping algebra UEA of a
$\theta$-colored ($\mathbb{Z}_{2} \times \mathbb{Z}_{2}$)-graded Lie algebra $L_{\mathbb{Z}_{2} \times \mathbb{Z}_{2}}$.
This implies that $P_{BF}$ is a ($\mathbb{Z}_{2} \times \mathbb{Z}_{2}$)-graded associative algebra
\begin{equation}
P_{BF} \cong \mathbb{U}(L_{\mathbb{Z}_{2} \times \mathbb{Z}_{2}})
\end{equation}
Its generators are homogeneous elements in the above gradation, with the paraboson generators $B_{k}^{+}$, $B_{l}^{-}$,
$k,l = 1, 2, ...$ spanning the $L_{10}$ subspace of $L_{\mathbb{Z}_{2} \times \mathbb{Z}_{2}}$, and the parafermion generators
$F_{\alpha}^{+}$, $F_{\beta}^{-}$, $\alpha, \beta = 1, 2, ...$ spanning the $L_{11}$ subspace of $L_{\mathbb{Z}_{2} \times \mathbb{Z}_{2}}$, thus
their grades are given as follows
\begin{equation} \label{gradPBF1}
\begin{array}{ccc}
deg(B_{k}^{\pm}) = (1,0) & & deg(F_{\alpha}^{\pm}) = (1,1)
\end{array}
\end{equation}
At the same time the polynomials $\{ B_{k}^{\epsilon}, B_{l}^{\eta} \}$ and $[ F_{\alpha}^{\epsilon}, F_{\beta}^{\eta} ]$ $\forall$
$k, l, \alpha, \beta = 1,2,... \ $ and $\forall$ $\epsilon, \eta = \pm$ span the subspace $L_{00}$ of $L_{\mathbb{Z}_{2} \times \mathbb{Z}_{2}}$,
and the polynomials $\{ F_{\alpha}^{\epsilon}, B_{k}^{\eta} \}$ $\forall$ $k, \alpha = 1,2,... \ $ and $\forall$ $\epsilon, \eta =
\pm$ span the subspace $L_{01}$ of $L_{\mathbb{Z}_{2} \times \mathbb{Z}_{2}}$. Consequently their grades are given as follows
\begin{equation}  \label{gradPBF2}
\begin{array}{c}
deg(\{ B_{k}^{+}, B_{l}^{+} \}) = deg(\{ B_{k}^{+}, B_{l}^{-} \}) = deg(\{ B_{k}^{-}, B_{l}^{-} \}) = (0,0) \\
   \\
deg([ F_{\alpha}^{+}, F_{\beta}^{+} ]) = deg([ F_{\alpha}^{+}, F_{\beta}^{-} ]) = deg([ F_{\alpha}^{-}, F_{\beta}^{-} ]) = (0,0) \\
  \\
deg(\{ B_{k}^{+}, F_{\alpha}^{+} \}) = deg(\{ B_{k}^{+}, F_{\alpha}^{-} \}) = deg(\{ F_{\alpha}^{+},
B_{k}^{-} \}) = deg(\{ B_{k}^{-}, F_{\alpha}^{-} \}) = (0,1)    \\
\end{array}
\end{equation}
Finally the color function used in the above construction is given
\begin{equation} \label{colfunctRelParSe}
\begin{array}{c}
\theta:\big( \mathbb{Z}_{2} \times \mathbb{Z}_{2} \big)
\times \big( \mathbb{Z}_{2} \times \mathbb{Z}_{2} \big)
\rightarrow \mathbb{C}^{*} \\
  \\
\theta(a,b) = (-1)^{(a_{1}b_{1} + a_{2}b_{2})}
\end{array}
\end{equation}
$\forall$ $a = (a_{1}, a_{2}), b = (b_{1}, b_{2}) \in \mathbb{Z}_{2} \times \mathbb{Z}_{2}$, and the operations in the exponent are considered in the
$\mathbb{Z}_{2}$ ring.
\end{theorem}

The UEA of a ($\mathbb{Z}_{2} \times \mathbb{Z}_{2}, \theta$)-Lie algebra $L$ is not a Hopf algebra, at least not in the ordinary sense. Instead it has
a structure which is strongly reminiscent of Hopf algebras: First we consider the ($\mathbb{Z}_{2} \times \mathbb{Z}_{2}$)-graded v.s.
$\mathbb{U}(L) \otimes \mathbb{U}(L)$, equipped with the multiplication
\begin{equation} \label{braidmultip}
(x \otimes y)(z \otimes w) = \theta(a,b) xz \otimes yw
\end{equation}
$\forall$ $x,w \in \mathbb{U}(L)$ and for $y,z$ homogeneous in $\mathbb{U}(L)$ with $deg(y) = a$, $deg(z) = b$, which becomes an associative
($\mathbb{Z}_{2} \times \mathbb{Z}_{2}$)-graded algebra and will be denoted $\mathbb{U}(L) \underline{\otimes} \mathbb{U}(L)$ and called
$\theta$-braided, ($\mathbb{Z}_{2} \times \mathbb{Z}_{2}$)-graded tensor product algebra from now on. Then $U(L)$ is equipped with a ``comultiplication''
\begin{equation} \label{brcopr}
\underline{\Delta} : \mathbb{U}(L) \rightarrow \mathbb{U}(L) \underline{\otimes} \mathbb{U}(L)
\end{equation}
which is a homomorphism of ($\mathbb{Z}_{2} \times \mathbb{Z}_{2}$)-graded algebras from $U(L)$ to the braided tensor product algebra
$\mathbb{U}(L) \underline{\otimes} \mathbb{U}(L)$ in the sense that
\begin{equation} \label{brcoprhom}
\underline{\Delta}(xy) = \sum \theta(deg(x_{2}), deg(y_{1})) x_{1}y_{1} \otimes x_{2}y_{2} = \underline{\Delta}(x) \underline{\Delta}(y)
\end{equation}
for any $x,y \in \mathbb{U}(L)$, where $\underline{\Delta}(x) = \sum x_{1} \otimes x_{2}$,
$\underline{\Delta}(y) = \sum y_{1} \otimes y_{2}$ and $x_{2}$, $y_{1}$ homogeneous with
$deg(x_{2}), deg(y_{1}) \in \mathbb{Z}_{2} \times \mathbb{Z}_{2}$.
The product $\underline{\Delta}(x) \underline{\Delta}(y)$ in the rhs of \eqref{brcoprhom} is to be understood in
$\mathbb{U}(L) \underline{\otimes} \mathbb{U}(L)$; thus it is given by \eqref{braidmultip}. $\underline{\Delta}$ is uniquely defined by its
values on the generators of $U(L)$ (i.e.: the basis elements of $L$)
\begin{equation} \label{brcoprval}
\underline{\Delta}(x) = 1 \otimes x + x \otimes 1
\end{equation}
Similarly, $\mathbb{U}(L)$ is equipped with an ``antipode'' $\underline{S} : U(L) \rightarrow U(L)$ which, unlike the case of ordinary Hopf algebras,
is not an algebra antihomomorphism but a ``twisted'' \cite{Mon} or braided antihomomorphism of ($\mathbb{Z}_{2} \times \mathbb{Z}_{2}$)-graded algebras
in the following sense
\begin{equation} \label{brantantihom}
\underline{S}(xy) = \theta(deg(x), deg(y)) \underline{S}(y)\underline{S}(x)
\end{equation}
for any homogeneous $x, y \in U(L)$. $S$ is uniquely defined by its values on the generators $U(L)$ (i.e.: the basis elements of $L$)
\begin{equation} \label{braantival}
\underline{S}(x) = -x
\end{equation}
If the above are supplemented with $\underline{\varepsilon}(x) = 0 $ $\ \forall$ $x \in \mathbb{U}(L)$ $\ x \neq 1$, $\underline{\varepsilon}(1) = 1$ then one can
straightforwardly check that all usual axioms of the Hopf structure (coassociativity, counity, antipode compatibility, etc) are satisfied.
The resulting algebraic structure is traditionally called in the
literature as a ($\mathbb{Z}_{2} \times \mathbb{Z}_{2}$)-graded Hopf algebra or a $(\mathbb{Z}_{2} \times \mathbb{Z}_{2}, \theta)$-graded Hopf algebra.
According to the modern terminology \cite{Maj,Mon}, developed in the '$90$'s and originating from Quantum Group theory, such a structure will  be called a $\theta$-braided group
in the sense of the braiding induced in the ${}_{\mathbb{C}(\mathbb{Z}_{2} \times \mathbb{Z}_{2})}\mathcal{M}$ category by the commutation factor $\theta$
(or by its equivalent description which will be the triangular structure of the $\mathbb{C}(\mathbb{Z}_{2} \times \mathbb{Z}_{2})$ group Hopf algebra).
It is also customary to speak of $U(L)$ as a Hopf algebra in the symmetric monoidal category ${}_{\mathbb{C}(\mathbb{Z}_{2} \times \mathbb{Z}_{2})}\mathcal{M}$
of representations of the group Hopf algebra $\mathbb{C}(\mathbb{Z}_{2} \times \mathbb{Z}_{2})$.

Based on the previous discussion and on the results of \thref{gradstrRelParabSet} we will conclude this section by stating and proving two Lemmas which will be of use in the
next section:
\begin{lemma} \lelabel{subspaces}
The linear subspace of the Relative Parabose set $P_{BF}$ (or of the ($\mathbb{Z}_{2} \times \mathbb{Z}_{2}, \theta$)-Lie algebra $L_{\mathbb{Z}_{2} \times \mathbb{Z}_{2}}$)
spanned by the elements of the form $\{ B_{k}^{\epsilon}, B_{l}^{\eta} \}$, $[ F_{\alpha}^{\epsilon}, F_{\beta}^{\eta} ]$ and $\{ F_{\alpha}^{\epsilon}, B_{k}^{\eta} \}$
$\forall$ $k, l, \alpha, \beta = 1,2,... \ $ and $\forall$ $\epsilon, \eta = \pm$ is a $\mathbb{Z}_{2}$-graded Lie algebra (or equivalently a Lie superalgebra). The UEA of this
Lie superalgebra is a subalgebra of $P_{BF}$.
\end{lemma}
\begin{proof}
From \thref{gradstrRelParabSet} we have that the elements $\{ B_{k}^{\epsilon}, B_{l}^{\eta} \}$, $[ F_{\alpha}^{\epsilon}, F_{\beta}^{\eta} ]$ and
$\{ F_{\alpha}^{\epsilon}, B_{k}^{\eta} \}$ $\forall$ $k, l, \alpha, \beta = 1,2,... \ $ and $\forall$ $\epsilon, \eta = \pm$ span the $L_{00} \oplus L_{01}$ subspace
of $L_{\mathbb{Z}_{2} \times \mathbb{Z}_{2}} = L_{00} \oplus L_{01} \oplus L_{10} \oplus L_{11}$.
Now it suffices to notice that the subset $\{(0,0), (0,1)\}$ of $\mathbb{Z}_{2} \times \mathbb{Z}_{2}$ is a subgroup isomorphic to the $\mathbb{Z}_{2}$ group and that the
restriction of the commutation factor \eqref{colfunctRelParSe} on $\{(0,0), (0,1)\}$ coincides (as a function) with the usual commutation factor of Lie superalgebras
(see \eqref{Liesupercom}). The fact that the corresponding UEA is a subalgebra of $P_{BF}$ is obvious by definition. Thus the proof is completed.

Alternatively, we could say that employing the commutation factor $\theta$ given in \eqref{colfunctRelParSe} and the gradation described in \thref{gradstrRelParabSet} we
get the following relations inside the relative parabose set
\begin{equation} \label{Z2subZ22}
\begin{array}{c}
L_{00} \ni \langle L_{00}, L_{00} \rangle =  [ L_{00}, L_{00} ]  \rightsquigarrow  \textrm{commutator} \\
    \\
L_{01} \ni \langle L_{00}, L_{01} \rangle =  [ L_{00}, L_{01} ]  \rightsquigarrow  \textrm{commutator}       \\
    \\
L_{00} \ni \langle L_{01}, L_{01} \rangle =  \{ L_{01}, L_{01} \}   \rightsquigarrow \textrm{anticommutator}      \\
\end{array}
\end{equation}
where $\langle .., .. \rangle : L \times L \rightarrow L$ is the non-associative multiplication of the ($\mathbb{Z}_{2} \times \mathbb{Z}_{2}, \theta$)-Lie algebra
$L_{\mathbb{Z}_{2} \times \mathbb{Z}_{2}}$
\end{proof}

We furthermore have the following
\begin{lemma} \lelabel{braidgrstrPbf}
The relative parabose set $P_{BF}$ has the structure of a $\theta$-braided group where the commutation factor $\theta:\big( \mathbb{Z}_{2} \times \mathbb{Z}_{2} \big)
\times \big( \mathbb{Z}_{2} \times \mathbb{Z}_{2} \big) \rightarrow \mathbb{C}^{*}$ is given by \eqref{colfunctRelParSe}. The relations can be given explicitly as
\small{\begin{equation} \label{brcombrantPbf}
\begin{array}{cc}
\underline{\Delta}(B_{i}^{\pm}) = 1 \otimes B_{i}^{\pm} + B_{i}^{\pm} \otimes 1  & \underline{S}(B_{i}^{\pm}) = - B_{i}^{\pm} \\
    \\
\underline{\Delta}(F_{j}^{\pm}) = 1 \otimes F_{j}^{\pm} + F_{j}^{\pm} \otimes 1 &  \underline{S}(F_{j}^{\pm}) = - F_{j}^{\pm} \\
    \\
\underline{\Delta}(\{ B_{k}^{\epsilon}, B_{l}^{\eta} \}) = 1 \otimes \{ B_{k}^{\epsilon}, B_{l}^{\eta} \} + \{ B_{k}^{\epsilon}, B_{l}^{\eta} \} \otimes 1 &
\underline{S}(\{ B_{k}^{\epsilon}, B_{l}^{\eta} \}) = - \{ B_{k}^{\epsilon}, B_{l}^{\eta} \}  \\
   \\
\underline{\Delta}([ F_{\alpha}^{\epsilon}, F_{\beta}^{\eta} ]) = 1 \otimes [ F_{\alpha}^{\epsilon}, F_{\beta}^{\eta} ] + [ F_{\alpha}^{\epsilon}, F_{\beta}^{\eta} ] \otimes 1  &
\underline{S}([ F_{\alpha}^{\epsilon}, F_{\beta}^{\eta} ]) = - [ F_{\alpha}^{\epsilon}, F_{\beta}^{\eta} ] \\
    \\
\underline{\Delta}(\{ F_{\alpha}^{\epsilon}, B_{k}^{\eta} \}) = 1 \otimes \{ F_{\alpha}^{\epsilon}, B_{k}^{\eta} \} + \{ F_{\alpha}^{\epsilon}, B_{k}^{\eta} \} \otimes 1   &
\underline{S}(\{ F_{\alpha}^{\epsilon}, B_{k}^{\eta} \}) = - \{ F_{\alpha}^{\epsilon}, B_{k}^{\eta} \}
\end{array}
\end{equation}}
$\forall$ $i, j, k, l, \alpha, \beta = 1,2,... \ $ and $\forall$ $\epsilon, \eta = \pm$. We also have $\underline{\varepsilon}(x) = 0 $ $\forall$ $x \in P_{BF}$.
\end{lemma}
\begin{proof}
This is a straightforward consequence of \thref{gradstrRelParabSet} and the succeeding discussion. For clarity we demonstrate the reproduction of one of the relations:
\small{$$
\begin{array}{c}
\underline{\Delta}(\{ F_{\alpha}^{\epsilon}, B_{k}^{\eta} \}) = \underline{\Delta}(F_{\alpha}^{\epsilon}B_{k}^{\eta} + B_{k}^{\eta}F_{\alpha}^{\epsilon}) =
\underline{\Delta}(F_{\alpha}^{\epsilon})\underline{\Delta}(B_{k}^{\eta}) + \underline{\Delta}(B_{k}^{\eta})\underline{\Delta}(F_{\alpha}^{\epsilon}) = \\
  \\
(1 \otimes F_{\alpha}^{\epsilon} + F_{\alpha}^{\epsilon} \otimes 1)(1 \otimes B_{k}^{\eta} + B_{k}^{\eta} \otimes 1) +
(1 \otimes B_{k}^{\eta} + B_{k}^{\eta} \otimes 1)(1 \otimes F_{\alpha}^{\epsilon} + F_{\alpha}^{\epsilon} \otimes 1) =  \\
   \\
1 \otimes F_{\alpha}^{\epsilon}B_{k}^{\eta} + \theta(deg(F_{\alpha}^{\epsilon}), deg(B_{k}^{\eta}))B_{k}^{\eta} \otimes F_{\alpha}^{\epsilon} + F_{\alpha}^{\epsilon} \otimes B_{k}^{\eta} + F_{\alpha}^{\epsilon}B_{k}^{\eta} \otimes 1 + \\
      \\
+ 1 \otimes B_{k}^{\eta}F_{\alpha}^{\epsilon} + \theta(deg(B_{k}^{\eta}), deg(F_{\alpha}^{\epsilon}))F_{\alpha}^{\epsilon} \otimes B_{k}^{\eta} + B_{k}^{\eta} \otimes F_{\alpha}^{\epsilon} + B_{k}^{\eta}F_{\alpha}^{\epsilon} \otimes 1 =    \\
    \\
= 1 \otimes \{ F_{\alpha}^{\epsilon}, B_{k}^{\eta} \} + \{ F_{\alpha}^{\epsilon}, B_{k}^{\eta} \} \otimes 1
\end{array}
$$}
since $\theta(deg(F_{\alpha}^{\epsilon}), deg(B_{k}^{\eta})) = \theta(deg(B_{k}^{\eta}), deg(F_{\alpha}^{\epsilon})) = -1$.
\end{proof}

\section{A new family of Realizations for an arbitrary Lie superalgebra} \selabel{newresults}

We recall that the prefix ``super'' will always amount to $\mathbb{Z}_{2}$-graded.

Let $L=L_{0} \oplus L_{1}$ be any complex Lie superalgebra of either finite or infinite dimension and let $V = V_{0} \oplus V_{1}$ be a finite
dimensional, complex, super-vector space i.e. $dim_{\mathbb{C}}V_{0} = m$ and $dim_{\mathbb{C}}V_{1} = n$. If $V$ is the carrier space for a
super-representation (or: a $\mathbb{Z}_{2}$-graded representation) of $L$, this is equivalent \cite{Cornw, Scheu2} to the existence of an homomorphism
$P: \mathbb{U}(L) \rightarrow End(V)$ of associative superalgebras. For any homogeneous element $z \in L$ the image $P(z)$ will be a  $(m+n)\times (m+n)$
matrix of the form
\begin{equation} \label{GRADEDREPR}
 P(z)= \left(\begin{array}{c|c}
              A(z) & B(z) \\ \hline
              C(z) & D(z)
             \end{array}\right)
\end{equation}
such that $B(z)=C(z)=0$ for all even elements and $A(z)=D(z)=0$ for all odd elements. Consequently, if $z \in L_{0}$ (equivalently if $z=X$ is an
even generator of $U(L)$) then \eqref{GRADEDREPR} becomes
\begin{equation} \label{evenmatr}
P(X)=
 \left(\begin{array}{ccc}
A(X) & 0 \\
0 & D(X)
\end{array}\right)
\end{equation}
while if $z \in L_{1}$ (equivalently if $z=Y$ is an odd generator of $U(L)$) then \eqref{GRADEDREPR} becomes
\begin{equation} \label{oddmatr}
P(Y) =
 \left(\begin{array}{ccc}
0 & B(Y) \\
C(Y) & 0
\end{array}\right)
\end{equation}
The submatrices $A_{m\times m}$, $B_{m\times n}$, $C_{n\times m}$,
$D_{n\times n}$,  of $P_{(m+n)\times (m+n)}$ constitute the
partitioning ($\mathbb{Z}_{2}$-grading) of the representation, with the dividing
lines in \eqref{GRADEDREPR} inserted in order to emphasize this
partitioning.

If $X_{i}$, $i = 1,2,...$ constitute an homogeneous basis in the even part $L_{0}$ of the Lie superalgebra $L$ and $Y_{j}$, $j = 1,2,...$
constitute an homogeneous basis of the odd part $L_{1}$ (thus: $X_{i}, \ Y_{j}$ $ \ \forall \ i = 1,2,... \ ,\ j = 1,2,...$ consitute a set of
homogeneous generators of $\mathbb{U}(L)$) then we have the following
\begin{theorem} \thlabel{parapartrealLiesuperalg}
The linear map $J_{P_{BF}}: L \rightarrow P_{BF}$ defined by
\begin{equation} \label{JordgenPBFeqeven}
\begin{array}{c}
 J_{P_{BF}} : L \longrightarrow P_{BF}     \\
                \\
   X_{i} \mapsto J_{P_{BF}}(X_{i}) = \frac{1}{2} \sum_{k,l=1}^{m}A_{kl}(X_{i}) \{ B_{k}^{+}, B_{l}^{-} \} +
\frac{1}{2} \sum_{\alpha,\beta=1}^{n}D_{\alpha\beta}(X_{i}) [ F_{\alpha}^{+}, F_{\beta}^{-} ] \\
\end{array}
\end{equation}
for any even element ($Z=X_{i}$) of an homogeneous basis of $L$ and by
\begin{equation} \label{JordgenPBFeqodd}
\begin{array}{c}
 J_{P_{BF}} : L \longrightarrow P_{BF}     \\
                \\
   Y_{j} \mapsto J_{P_{BF}}(Y_{j}) =
  \frac{1}{2} \sum_{k=1}^{m}\sum_{\alpha=1}^{n} \Big( B_{k,\alpha}(Y_{j}) \{ B_{k}^{+},
   F_{\alpha}^{-} \}
+ C_{\alpha,k}(Y_{j}) \{ F_{\alpha}^{+}, B_{k}^{-} \} \Big) \\
\end{array}
\end{equation}
for any odd element ($Z=Y_{j}$) of an homogeneous basis of $L$, can be extended to an homomorphism of associative algebras
$J_{P_{BF}}: \mathbb{U}(L) \rightarrow P_{BF}$ between the universal enveloping algebra $\mathbb{U}(L)$ of the Lie superalgebra $L$
and the relative parabose set $P_{BF}$, in other words it is a realization of $L$ with paraparticles.

Furthermore, the above constructed homomorphism of associative algebras $J_{P_{BF}}: \mathbb{U}(L) \rightarrow P_{BF}$, is an homomorphism
of super-Hopf algebras (equivalently an homomorphism of $\mathbb{Z}_{2}$-graded Hopf algebras) between $\mathbb{U}(L)$ and a suitable
$\mathbb{Z}_{2}$-graded subalgebra of $P_{BF}$
\end{theorem}
\begin{proof}
The non-associative multiplication $\langle .., .. \rangle : L \times L \rightarrow L$ of the Lie superalgebra $L$,
when specified in the even and odd subspaces can be written (inside the $\mathbb{U}(L)$) in detail as follows
\begin{equation} \label{strconstsup}
\begin{array}{c}
\langle  X_{i}, X_{j}  \rangle = \big[ X_{i}, X_{j} \big] = \sum c_{ij}^{k}X_{k} \\
   \\
\langle  X_{i}, Y_{p}  \rangle = \big[ X_{i}, Y_{p} \big] = \sum d_{ip}^{k}Y_{k}  \\
     \\
\langle  Y_{p}, Y_{q}  \rangle = \{ Y_{p}, Y_{q} \} = \sum f_{pq}^{k}X_{k} \\
\end{array}
\end{equation}
for $i, j, p, q = 1,2,...$. The complex numbers $c_{ij}^{k}, d_{ip}^{k}, f_{pq}^{k}$ are usually called the structure constants of $L$
and there is a finite number of terms in the sums of the rhs.

For the first part of the theorem, it is sufficient to show that
\begin{equation} \label{verifparapreal}
\begin{array}{c}
J_{P_{BF}}(\big[ X_{i}, X_{j} \big]) = \big[ J_{P_{BF}}(X_{i}), J_{P_{BF}}(X_{j}) \big]   \\
    \\
J_{P_{BF}}(\big[ X_{i}, Y_{p} \big]) = \big[ J_{P_{BF}}(X_{i}), J_{P_{BF}}(Y_{p}) \big]  \\
   \\
J_{P_{BF}}(\{ Y_{p}, Y_{q} \}) = \{ J_{P_{BF}}(Y_{p}), J_{P_{BF}}(Y_{q}) \}  \\
\end{array}
\end{equation}
for all values of $i, j, p, q = 1,2,...$   \\

If we replace in the rhs of the equations \eqref{verifparapreal} the relations \eqref{JordgenPBFeqeven}, \eqref{JordgenPBFeqodd}, we get relations with mixed quadruplex
commutators and anticommutators. For example the rhs of the first of the equations \eqref{verifparapreal} becomes
\begin{equation} \label{verifparapreal1}
\begin{array}{c}
\big[ J_{P_{BF}}(X_{i}), J_{P_{BF}}(X_{j}) \big] =  \\
         \\
  \frac{1}{4} \sum_{k,l=1}^{m} \sum_{k',l'=1}^{m} A_{kl}(X_{i})
A_{k'l'}(X_{j}) \big[ \ \{ B_{k}^{+}, B_{l}^{-} \}, \{ B_{k'}^{+},
B_{l'}^{-} \} \
\big] +  \\
    \\
  + \frac{1}{4} \sum_{k,l=1}^{m} \sum_{\alpha',\beta'=1}^{n}
A_{kl}(X_{i}) D_{\alpha' \beta'}(X_{j}) \big[ \ \{ B_{k}^{+},
B_{l}^{-} \},
\big[ \ F_{\alpha'}^{+}, F_{\beta'}^{-} \big] \ \big] +   \\
   \\
  + \frac{1}{4} \sum_{k',l'=1}^{m} \sum_{\alpha,\beta=1}^{n}
D_{\alpha \beta}(X_{i}) A_{k'l'}(X_{j})
 \big[ \ \big[
\ F_{\alpha}^{+}, F_{\beta}^{-} \big], \{ B_{k'}^{+}, B_{l'}^{-}
\}
\ \big] +   \\
     \\
  + \frac{1}{4} \sum_{\alpha,\beta=1}^{n}
\sum_{\alpha',\beta'=1}^{n} D_{\alpha \beta}(X_{i}) D_{\alpha'
\beta'}(X_{j})
 \big[ \ \big[
\ F_{\alpha}^{+}, F_{\beta}^{-} \big], \big[ F_{\alpha'}^{+},
F_{\beta'}^{-} \big] \ \big]   \\
\end{array}
\end{equation}
Now we proceed by ``breaking'' the outer commutators in the rhs of the above relation using the following identities
\begin{equation} \label{genassocidentit}
\begin{array}{c}
\{ \ \{ A_{1} , A_{2} \} , \{ A_{3} , A_{4} \} \ \} = \{ A_{1},
\{A_{2}, \{A_{3}, A_{4} \} \ \} \ \} + \big[ A_{2}, \big[A_{1},
\{A_{3}, A_{4} \} \ \big] \ \big] \\
  \\
\big[ \ \{ A_{1} , A_{2} \} , \big[ A_{3} , A_{4} \big] \ \big] =
\{ A_{1}, \big[A_{2}, \big[A_{3}, A_{4} \big] \ \big] \ \} + \{
A_{2},
\big[A_{1}, \big[A_{3}, A_{4} \big] \ \big] \ \}  \\
   \\
\big[ \ \{ A_{1} , A_{2} \} , \{ A_{3} , A_{4} \} \ \big] = \{
A_{1}, \big[A_{2}, \{A_{3}, A_{4} \} \ \big] \ \} + \{ A_{2},
\big[A_{1},
\{A_{3}, A_{4} \} \ \big] \ \} \\
  \\
\big[ \ \big[ A_{1} , A_{2} \big] , \big[ A_{3} , A_{4} \big]
\big] = \big[ A_{1}, \big[A_{2}, \big[A_{3}, A_{4} \big] \ \big] \
\big] - \big[ A_{2}, \big[A_{1}, \big[A_{3}, A_{4} \big] \ \big] \
\big]
\end{array}
\end{equation}
Relations \eqref{genassocidentit} hold identically for any (associative) variables. Applying these identities in the rhs of
\eqref{verifparapreal1} results in trilinear expressions mixing parabosonic and parafermionic degrees of freedom.
Now we apply on these resulting trilinear expressions the equations \eqref{parab-paraf}, \eqref{parabparaf}
which describe the algebra of the relative parabose set.
After lengthy algebraic calculations, and taking into account the properties of the submatrices $A_{m\times m}$, $B_{m\times n}$, $C_{n\times m}$,
$D_{n\times n}$ of the graded matrix representation \eqref{GRADEDREPR}, \eqref{evenmatr}, \eqref{oddmatr} we end up in the lhs \eqref{verifparapreal}.
This completes the first part of the proof, i.e. the fact that the map described by \eqref{JordgenPBFeqeven}, \eqref{JordgenPBFeqodd} is a
realization for the Lie superalgebra.

According now to the second part of the theorem, i.e. the fact that the constructed realization is an homomorphism of super-Hopf algebras
between $\mathbb{U}(L)$ and a suitable $\mathbb{Z}_{2}$-graded subalgebra of $P_{BF}$ we have the following: From the results of \thref{gradstrRelParabSet}, we know
that $\{ B_{k}^{\epsilon}, B_{l}^{\eta} \}$ and $[ F_{\alpha}^{\epsilon}, F_{\beta}^{\eta} ]$ $\forall$
$k, l, \alpha, \beta = 1,2,... \ $ and $\forall$ $\epsilon, \eta = \pm$ span the subspace $L_{00}$ of $L_{\mathbb{Z}_{2} \times \mathbb{Z}_{2}}$,
and $\{ F_{\alpha}^{\epsilon}, B_{k}^{\eta} \}$ $\forall$ $k, \alpha = 1,2,... \ $ and $\forall$ $\epsilon, \eta =
\pm$ span the subspace $L_{01}$ of $L_{\mathbb{Z}_{2} \times \mathbb{Z}_{2}}$. Consequently we have
\begin{equation}  \label{gradPBF2}
\begin{array}{c}
deg(\{ B_{k}^{\epsilon}, B_{l}^{\eta} \}) = deg([ F_{\alpha}^{\epsilon}, F_{\beta}^{\eta} ]) = (0,0)   \\
  \\
deg(\{ F_{\alpha}^{\epsilon}, B_{k}^{\eta} \}) = (0,1)    \\
\end{array}
\end{equation}
But from \leref{subspaces} we know that $L_{00} \oplus L_{01}$ is a Lie superalgebra which is a subalgebra of $L_{\mathbb{Z}_{2} \times \mathbb{Z}_{2}}$.
$L_{00}$ is the even subspace and $L_{01}$ is the odd subspace. Thus, the images $J_{P_{BF}}(X)$ and $J_{P_{BF}}(Y)$ (where: $X$ and $Y$ the even and odd
elements of an homogeneous basis of the arbitrary Lie superalgebra $L$) of the even linear map $J_{P_{BF}}: L \rightarrow P_{BF}$ defined by equations
\eqref{JordgenPBFeqeven}, \eqref{JordgenPBFeqodd}, belong in the $L_{00} \oplus L_{01}$ subalgebra of $L_{\mathbb{Z}_{2} \times \mathbb{Z}_{2}}$. Consequently
$$
J_{P_{BF}}: \mathbb{U}(L) \rightarrow \mathbb{U}(L_{00} \oplus L_{01}) \subset P_{BF}
$$
is an homomorphism of associative superalgebras. Now it suffices to show that the following diagrams are commutative
\small{\begin{equation}
\xymatrix{  \mathbb{U}(L) \ar[r]^{J_{P_{BF}}} \ar[d]_{\Delta_{L}} & \mathbb{U}(L_{00} \oplus L_{01}) \ar[d]^{\underline{\Delta}}\\
            \mathbb{U}(L) \otimes \mathbb{U}(L) \ar[r]^{J_{P_{BF}} \otimes J_{P_{BF}}}  & \mathbb{U}(L_{00} \oplus L_{01}) \otimes \mathbb{U}(L_{00} \oplus L_{01})       }
\end{equation}}
and
\small{\begin{equation}
\begin{array}{ccc}
\xymatrix{ \mathbb{U}(L) \ar[rr]^{J_{P_{BF}}} \ar[dr]_{\varepsilon_{L}} & & \mathbb{U}(L_{00} \oplus L_{01}) \ar[dl]^{\underline{\varepsilon}} \\
               &\mathbb{C}& }
                        & &
\xymatrix{  \mathbb{U}(L) \ar[r]^{J_{P_{BF}}} \ar[d]_{S_{L}} & \mathbb{U}(L_{00} \oplus L_{01}) \ar[d]^{\underline{S}}\\
            \mathbb{U}(L)  \ar[r]^{J_{P_{BF}}}  & \mathbb{U}(L_{00} \oplus L_{01})  }
                         \\
\end{array}
\end{equation}}
or equivalently that the following relations are valid $\forall$ $x \in L$ (the initial Lie superalgebra)
\begin{equation}
\begin{array}{c}
\underline{\Delta} \circ J_{P_{BF}} = (J_{P_{BF}} \otimes J_{P_{BF}}) \circ \Delta_{L}  \\
   \\
\underline{\varepsilon} \circ J_{P_{BF}} = \varepsilon_{L} \\
  \\
\underline{S} \circ J_{P_{BF}} = J_{P_{BF}} \circ S_{L}
\end{array}
\end{equation}
The above  relations can be shown with straightforward calculations on the (homogeneous) generators of the Lie superalgebra $L$.
In the above we denote $\underline{\Delta}$, $\underline{\varepsilon}$ και
$\underline{S}$, the restrictions of the corresponding maps \eqref{brcombrantPbf} of $P_{BF}$ on its $\mathbb{Z}_{2}$-graded subalgebra
$\mathbb{U}(L_{00} \oplus L_{01})$, while $\Delta_{L}$, $\varepsilon_{L}$ and $S_{L}$ are the super-Hopf algebra structure maps of the UEA $\mathbb{U}(L)$
of the initial Lie superalgebra $L$. Thus the proof is complete
\end{proof}

We are now going to present some corollaries of the above theorem, which will indicate its connections with previous results found in the literature.
Let $L$ be any Lie algebra of either finite or infinite dimension. Let $V$, $W$ two finite dimensional vector spaces of dimensions $m$ and $n$
respectively, such that $A : L \rightarrow End(V)$, $D : L \rightarrow End(W)$ are two matrix representations of $L$ i.e. for any $x \in L$, $\ A(x) \ $, $D(x)$
are $m \times m$ and $n \times n$ dimensional matrices respectively. Restricting ourselves to the even part of the previous theorem we get the following
immediate consequence
\begin{corollary}  \colabel{Liealgreal1}
The linear map defined by
\begin{equation} \label{JordgenPBFLieeq}
\begin{array}{c}
 J_{P_{BF}} : L \longrightarrow P_{BF}     \\
                \\
   X \mapsto J_{P_{BF}}(X) = \frac{1}{2} \sum_{k,l=1}^{m}A_{kl}(X) \{ B_{k}^{+}, B_{l}^{-} \} +
\frac{1}{2} \sum_{\alpha,\beta=1}^{n}D_{\alpha\beta}(X) [ F_{\alpha}^{+}, F_{\beta}^{-} ] \\
\end{array}
\end{equation}
for any element of $L$, can be extended to an (ordinary) Hopf algebra homomorphism
$$
J_{P_{BF}}: \mathbb{U}(L) \rightarrow \mathbb{U}(L_{00}) \subset P_{BF}
$$
between the UEA $\mathbb{U}(L)$  of the Lie algebra $L$ and the UEA $\mathbb{U}(L_{00})$
which is a subalgebra of the relative parabose set $P_{BF}$. In other words \eqref{JordgenPBFLieeq} constitutes a paraparticle realization of an arbitrary
Lie algebra.
\end{corollary}
The next two are direct consequences of the above
\begin{corollary}  \colabel{Liealgreal2}
The linear map defined by
\begin{equation} \label{JordgenPBLieeq}
\begin{array}{c}
 J_{P_{B}} : L \longrightarrow P_{B}     \\
                \\
   X \mapsto J_{P_{B}}(X) = \frac{1}{2} \sum_{k,l=1}^{m}A_{kl}(X) \{ B_{k}^{+}, B_{l}^{-} \}  \\
\end{array}
\end{equation}
for any element of $L$, can be extended to an (ordinary) Hopf algebra homomorphism
\begin{equation} \label{lastmin}
J_{P_{B}}: \mathbb{U}(L) \rightarrow P_{B}
\end{equation}
between the UEA $\mathbb{U}(L)$  of the Lie algebra $L$ and the parabosonic algebra
$P_{B}$. In other words \eqref{JordgenPBLieeq} constitutes a paraboson realization of an arbitrary Lie algebra.
\end{corollary}
In the above, the paraboson algebra $P_{B}$ is defined by generators and relations by the first of the relations \eqref{parab-paraf}. This algebra is well known to be
a super-Hopf algebra \cite{KaDa}, and it is easy to see that the images of any $x \in \mathbb{U}(L)$ through \eqref{lastmin}, lie entirely in its even subspace, which
constitutes an ordinary Hopf subalgebra.
\begin{corollary}  \colabel{Liealgreal3}
The linear map defined by
\begin{equation} \label{JordgenPFLieeq}
\begin{array}{c}
 J_{P_{F}} : L \longrightarrow P_{F}     \\
                \\
   X \mapsto J_{P_{F}}(X) = \frac{1}{2} \sum_{\alpha,\beta=1}^{n}D_{\alpha\beta}(X) [ F_{\alpha}^{+}, F_{\beta}^{-} ] \\
\end{array}
\end{equation}
for any element of $L$, can be extended to an (ordinary) Hopf algebra homomorphism
$$
J_{P_{F}}: \mathbb{U}(L) \rightarrow P_{F}
$$
between the UEA $\mathbb{U}(L)$  of the Lie algebra $L$ and the parafermionic algebra $P_{F}$.
In other words \eqref{JordgenPFLieeq} constitutes a parafermion realization of an arbitrary Lie algebra.
\end{corollary}
In the above, the parafermion algebra $P_{F}$ is defined by generators and relations by the second of the relations \eqref{parab-paraf}. This algebra is well known
\cite{KaDa2} to be an ordinary Hopf algebra.

\coref{Liealgreal2} and \coref{Liealgreal3}  constitute generalizations of previous results of the first two authors of this paper,
on paraparticle realizations of finite dimensional Lie algebras: If we consider a finite dimensional Lie algebra and in place
of the matrix representations $A$, $D$ we take the adjoint representation of the Lie algebra (i.e.: the matrix elements will be the structure constants), then
the corresponding Lie algebra paraparticle realizations of \cite{DaKa} are immediately reproduced as special cases of \coref{Liealgreal2} and \coref{Liealgreal3} respectively.
However these corollaries provide a much richer class of realizations and are applicable for infinite dimensional Lie algebras as well.

On the other hand, the realizations provided at \thref{parapartrealLiesuperalg} can be considered as generalizing the Lie superalgebra particle realizations
provided in \cite{Pal17, ChangS}: The usual bosonic and fermionic number operators used in these works
are replaced by parabosonic and parafermionic number operators in \thref{parapartrealLiesuperalg}. Furthermore our use of a mixed paraparticle system which was
shown in \leref{subspaces} to have the structure of a braided group, enables the realization to be extended as an homomorphism of braided groups or
more specifically as an homomorphism of super-Hopf algebras, a feature which is not achieved in works as \cite{Pal17, ChangS} where the target algebra
consists of particles (bosons and fermions), since a braided group structure (or any other kind of Hopf structure) is not available for such particle algebras.

\section{Discussion - Prospects}

An extended discussion of the $\mathbb{Z}_{2} \times \mathbb{Z}_{2}$-graded structure and the properties of the Relative parabose set $P_{BF}$ has been provided
in \seref{relparabset}, together with a rigorous formulation of previous results. We have extracted the $\theta$-braided group structure of $P_{BF}$ as a consequence of this
discussion and we have investigated some of its subalgebras.

The results of \seref{newresults} provide a totally new class of realizations for an arbitrary Lie superalgebra of either finite or infinite dimension. These
are paraparticle realizations and they generalize previous results of the first two authors of this paper \cite{DaKa}
(on paraparticle realizations of finite dimensional Lie algebras) and of other authors \cite{Pal17, ChangS} (on particle realizations of Lie superalgebras) as well.
The fact that the paraparticle realizations constructed in this paper have the strong property of being homomorphisms of super-Hopf algebras has important consequences:
The powerful technique of creating tensor product representations is at our disposal. The (braided) tensor products of representations of the target algebra can be
directly utilized to create new representations of the Lie superalgebra. Consequently, realizations possessing the property of relating the Hopf structure of the algebras
as well (and not only their multiplicative structure) provide more solid links between their representation theories.
It is worth noticing at this point that the relative parabose set has been used once, in \cite{BiswSo} (although without being named)
in the construction of a specific realization of the finite dimensional orthosymplectic Lie superalgebra. However in this work, there is no mention of wether
the obtained map is a super-Hopf homomorphism or not.   \\

Before closing this work, let us present a rough outline of some ideas on virtual applications of the above presented constructions.
We will discuss about two classes of possible applications, which however lie in totally different directions of research:
The first one has to do with pure mathematics and more specifically with the representation theory of the infinite dimensional Lie (super)algebras, while the second one
stems from some attempts, of the first and the third authors of this paper, to determine a specific class of solutions of a model widely used in elementary particle and
nuclear physics as well, the so-called Skyrme Model \cite{Sky}.

The relative parabose set algebra has a totally unexplored representation theory. The Fock-like representations of such an algebra have been defined in \cite{GreeMe}
but have never been explicitly constructed due to obvious computational difficulties introduced by the nature of the ternary relations. In recent works \cite{KaDa3},
\cite{LiStVdJ, Stoil} different ideas for confronting such difficulties have been introduced for the case of the parabosonic algebras. Such ideas can be generalized for the case
of the relative parabose set, and possibly lead to the construction of its Fock-like representations. In such a case, \thref{parapartrealLiesuperalg} will be a
powerful tool which will enable us to proceed in developments in the representation theory of Lie superalgebras, a topic which
-at least in the infinite dimensional case- is virtually unexplored.

The $SU(N)$, $(N \geq 2)$ Skyrme Model \cite{Sky} has proved to be one of the most effective Lagrangians
for the description of strong interactions at the low energy limit of Quantum Chromodynamics (QCD). However its Euler-Lagrange equations are strongly non-linear,
and despite the wealth of physics they describe, have never been solved analytically (at least in the form of some general solution). In \cite{AKP1},
the first and the third authors of this paper, introduce a method which allows the detection of non-linear symmetries of the model.
In the same work, necessary conditions are computed, under which such symmetries may lead to further integration of the equations of
the model. One of our main initial motivations for studying (super)algebra realizations via paraparticles, stems from trying to extend the ideas of \cite{AKP1}, via imposing
``parafermionic'' conditions of the form
\begin{equation} \label{parafermionlikealg}
[[L_{\mu}, L_{\nu}], L_{\rho}] = f_{\mu\nu\rho}^{\sigma}L_{\sigma}
\end{equation}
($f_{\mu\nu\rho}^{\sigma} \in \mathbb{R} \ or \ \mathbb{C}$) between the Maurer-Cartan covariant vectors $L_{\mu}$, of the Euler-Lagrange equations of the Skyrme Model,
enforcing them to be the generators of a parafermion-like algebra. ``Realizing'' the Maurer-Cartan covariant vectors $L_{\mu}$ in terms of some paraparticle algebra,
may prove a helpful intermediate step in determining the constants $f_{\mu\nu\rho}^{\sigma}$. The next step will be the computation of the restrictions imposed on the forms of the fields
$U(x)$, implied by the conditions \eqref{parafermionlikealg}. Finally, if the produced forms will prove acceptable, we will proceed in computing the representations of the
parafermion-like generated algebra and the physics produced on these representations.

All these are parts of different ongoing projects and will constitute the subject of future works.

{\footnotesize\section*{Acknowledgements}

KK would like to express his gratitude to the whole Institute of Physics and Mathematics (IFM), of the University of Michoacan (UMSNH), for providing a stimulating
atmosphere while this work was in progress. He is also grateful for a grant for postdoctoral studies provided by \textsc{Conacyt} under the Project No. J60060.
The research of AHA was supported by grants CIC–4.16 and CONACYT J60060; he is also grateful to SNI.}

\end{document}